\documentclass[conference,10pt]{IEEEtran}

\usepackage{cite}

\ifCLASSINFOpdf
   \usepackage[pdftex]{graphicx}

\else

\fi

\bstctlcite{IEEEexample:BSTcontrol}

\usepackage{cite}

\usepackage{amsmath,amssymb,amsthm}

\usepackage{mathtools}
\usepackage{dblfloatfix}

\usepackage{mathrsfs}
\newtheorem{theorem}{Theorem}

\newtheorem{lemma}[theorem]{Lemma}
\newtheorem{claim}{Claim}
\newtheorem{corollary}[theorem]{Corollary}

\begin{document}

\title{On Achievability of an $(r,l)$ Fractional Linear Network Code}

\author{\IEEEauthorblockN{Niladri Das and Brijesh Kumar Rai}
\IEEEauthorblockA{Department of Electronics and Electrical Engineering, Indian Institute of Technology Guwahati, Guwahati, Assam, India\\
Email: \{d.niladri,bkrai\}@iitg.ernet.in}}

\maketitle

\begin{abstract}
It is known that there exists a network, called as the M-network, which is not scalar linearly solvable but has a vector linear solution for message dimension two. Recently, a generalization of this result has been presented where it has been shown that for any integer $m\geq 2$, there exists a network which has a $(m,m)$ vector linear solution, but does not have a $(w,w)$ vector linear solution for $w<m$. This paper presents a further generalization. Specifically, we show that for any positive integers $k,n,$ and $m\geq 2$, there exists a network  which has a $(mk,mn)$ fractional linear solution, but does not have a $(wk,wn)$ fractional linear solution for $w<m$.
\end{abstract}

\IEEEpeerreviewmaketitle

\section{Introduction}
The concept of network coding emerged in year 2000 where it was shown that by performing operations on incoming packets and then forwarding by nodes in a network may increase the throughput significantly\cite{alshwede}. Since then there have been extensive studies on various theoretical aspects and applications of network coding. The benefits of network coding can be best seen in a multicast network where the  min-cut bound can be achieved using network coding which otherwise may be impossible through only routing \cite{alshwede}.

Linear operations by the nodes is the most natural setting to consider due to the mathematical traceability of linear operations. The network coding involving only linear operations over a finite field is referred as the linear network coding. An $(r,l)$ fractional linear network code considers the information generated by the sources as $r$ length vectors of symbols and the information passing through all the edges as $l$ length vectors of symbols. A network is said to have an $(r,l)$ fractional linear network code solution if the sources can supply $r$ symbols to their respective terminals in $l$ uses of the network. If $r = l$, then an $(r,l)$ fractional linear network code is referred to as an $(r,r)$ vector linear network code (also as an $r$-dimensional vector linear code). An $1$-dimensional vector linear network code is referred to as a scalar linear network code. An $(r,l)$ factional linear network code for a network is achievable if the network has an $(r,l)$ fractional linear network code solution. A rate $r/l$ linear network code is said to be achievable if there exists an $(mr,ml)$ factional linear network code solution over some finite field where $m$ is a positive rational number such that $mr$ and $ml$ are integers. 

In this paper, we show that achievability of an arbitrary rate $r/l$ linear network code does not necessarily imply the existence of an $(r,l)$ fractional linear network code solution for any positive integers $r$ and $l$. Specifically, we show that for any integer $m\geq 2$, there exists a network where a rate $r/l$ linear network code is achievable but it does not have a $(wr,wl)$ fractional linear network code solution for any $w$ less than $m$. For $r=l$, the result of this paper specializes to the result of \cite{das}.

\subsection{Related Work}
In this subsection, we provide a brief survey of the studies related to the work of this paper. Scalar linear network codes over sufficiently large finite fields have been shown to be sufficient to achieve the capacity for  multicast networks \cite{li}. It has also been shown that such scalar linear codes can be efficiently designed \cite{jaggi}. Moreover, capacity achieving vector linear network codes can also be designed where the requirement of field size can be reduced as compared to capacity achieving scalar linear network codes\cite{ebra}. It has been shown in \cite{riis} and \cite{koetter} that all multicast networks are vector linearly solvable over the binary field $\mathbb{F}_2$ if the vector length is sufficiently large. For the multicast networks, it is also known that solution over an alphabet might not guarantee a solution over all larger alphabets \cite{doug2}. Although multicast networks have been extensively studied since the advent of network coding and there are many nice results such as characterising the capacity, dependency on the size of the alphabet etc., yet not everything is known and new results such as \cite{sun} continue to enhance the knowledge about multicast networks. Sun \textit{et al.} \cite{sun} presented an instance of a multicast network which has an $L$-dimensional vector linear network code solution over $\mathbb{F}_q$ but has no scalar linear network code solution over $\mathbb{F}_{q^\prime}$ where $q^\prime \leq q^L$. They also presented an explicit instance to show that the existence of a vector linear network code solution for a certain message dimension does not necessarily guarantee a vector linear network code solution for all higher message dimensions.

For non-multicast networks, the results are fundamentally different from multicast networks. Linear network codes are shown to be insufficient \cite{doug}. Among linearly solvable networks, there are networks which are vector linearly solvable only for certain vector dimensions. In \cite{medard}, it has been shown that there exists a network which has a 2-dimensional vector linear network code solution but has no scalar linear network code solution. This network is known as the M-network. In fact, the M-network has no vector linear solution for any odd message dimensions \cite{dougherty}. The results of \cite{medard} and \cite{dougherty} have been generalized in \cite{das}. In \cite{das}, it has been shown that for any integer $m\geq 2$, there exists a network which has an $m$-dimensional vector linear solution but has no $w$-dimensional vector linear solution for $w<m$.

Matroid theory, a branch of mathematics, has been shown to be useful to study the linear solvability of networks. Indeed, some of the important results, such as \cite{dougherty,doug,doug2,doug3,doug4,doug5,sun}, have been obtained by exploiting the connection between matroids and networks. It has been shown that a network has a scalar linear network code solution over $\mathbb{F}_q$ if and only if it is a matroidal network with respect to a matroid representable over $\mathbb{F}_q$ \cite{dougherty,kim}. Using this result the authors of \cite{doug} constructed the Fano network from the Fano matroid to show that there exists a network which has a vector linear solution for any message dimension over any finite field of even characteristics but has no vector linear solution for any message dimension when the characteristics of the finite field is odd. Similarly, from the the non-Fano matroid the non-Fano network was constructed to show that there exists a network which has a vector linear solution for any message dimension over a finite field of odd characteristics but has no vector linear solution for any message dimension when the characteristics of the finite field is even. In \cite{dougherty} the V{\'{a}}mos network was constructed from the V{\'{a}}mos matroid and it was shown that Shannon-type inequalities are insufficient to compute the network coding capacity of the V{\'{a}}mos network. \cite{dougherty} was generalized in \cite{sundar} to show that similar to the correspondence between matroids and scalar linear network code solution of a network, discrete polymatroids corresponds to both vector linear network code solution and fractional linear network code solution of a network. It has been shown that a network has an $(r,l)$ fractional linear network code solution over $\mathbb{F}_q$ if and only if it is a $(r,l)$-discrete polymatroidal network with respect to a discrete polymatroid representable over $\mathbb{F}_q$ \cite{sundar}.

\subsection{Organization of the paper}
The organization of the paper is as follows. Section~\ref{prelims} presents the formal definitions related to network coding. In Section~\ref{sec2}, we present a key result (Theorem \ref{thm1}): for any integer $m\geq 2$, there exists a network which has an $(m,mn)$ fractional linear network code solution but has no $(w,wn)$ fractional linear network code solution for $w<m$. Using this result, we present the main result of the paper (Theorem \ref{thm2}): there exists a network which has an $(mk,mn)$ fractional linear network code solution but does not have a $(wk,wn)$ fractional linear network code solution for $w<m$. The proof of Theorem~\ref{thm1} is presented in Section~\ref{sec4}. Section~\ref{conclusion} concludes the paper.

\section{Preliminaries}\label{prelims}
We represent a network by a directed acyclic graph. The sources of a network are represented by the nodes in the graph with no incoming edges; and the terminals of a network are represented by the nodes in the graph with no outgoing edges. The sources are assumed to generate an i.i.d. random process uniformly distributed over over $\mathbb{F}_q$. The random process at one source is independent of any collection of the random processes generated at all other sources. Each edge in the graph is assumed to be of unit capacity. For any node $v$, the set of edges incoming to a node $v$ and the set of edges outgoing from node $v$ are denoted by $In(v)$ and $Out(v)$ respectively. $(u,v)$ represents an edge $e$ directed from $u$ to $v$. Each terminal demands a set of source processes. Let the set of sources and the set of terminals be denoted by $S$ and $T$ respectively.

An $(r,l)$ fractional linear network code is defined as follows. In such a code, $r$ symbols are considered at every source. Let the source process generated by the source $s_i\in S$ be denoted by $X_i\in \mathbb{F}_q^r$. Each edge is used for $l$ units of time. Let the message symbol carried by an edge $e$ be denoted by $Y_e$ where $Y_e \in \mathbb{F}_q^l$. For an edge $e$, if $tail(e)=s_i$, then $Y_e = A_{\{i,e\}}X_i$ where $A_{\{i,e\}}\in \mathbb{F}_q^{l\times r}$. Else, if $tail(e)$ is an intermediate node, then for $Y_e = \sum_{e^\prime \in In(tail(e))} A_{\{e^\prime,e\}}Y_e^\prime$ where $A_{\{e^\prime,e\}} \in \mathbb{F}_q^{l\times l}$. Matrices $A_{\{i,e\}}$ and $A_{\{e^\prime,e\}}$ are called the local coding matrices. The terminals compute a set of $r$ length vectors from its incoming edges. If $X_{t_i}\in\mathbb{F}_q^r$ denotes a vector that is computed by a terminal $t_i\in T$ then, $X_{t_i} =  \sum_{e\in In(t_i)} A_{\{e,t_i\}}Y_{e}$ where $A_{\{e,t_i\}}\in \mathbb{F}_q^{r\times l}$. A network is said to have an $(r,l)$ fractional linear network code solution if the sources are able to send its symbols to the respective terminals in $l$ usages of the network using an $(r,l)$ fractional linear network code. An $(r,l)$ fractional linear network code is said to be achievable if there exists an $(r,l)$ fractional linear network code solution.

A network is said to have an $(r,l)$ routing solution if the sources are able to send its symbols to the respective terminals in $l$ usages of the network only through routing.

\section{Fractional Linear Network Code Solutions} \label{sec2}
\begin{figure}[!t]
\centering
\includegraphics[width=0.48\textwidth]{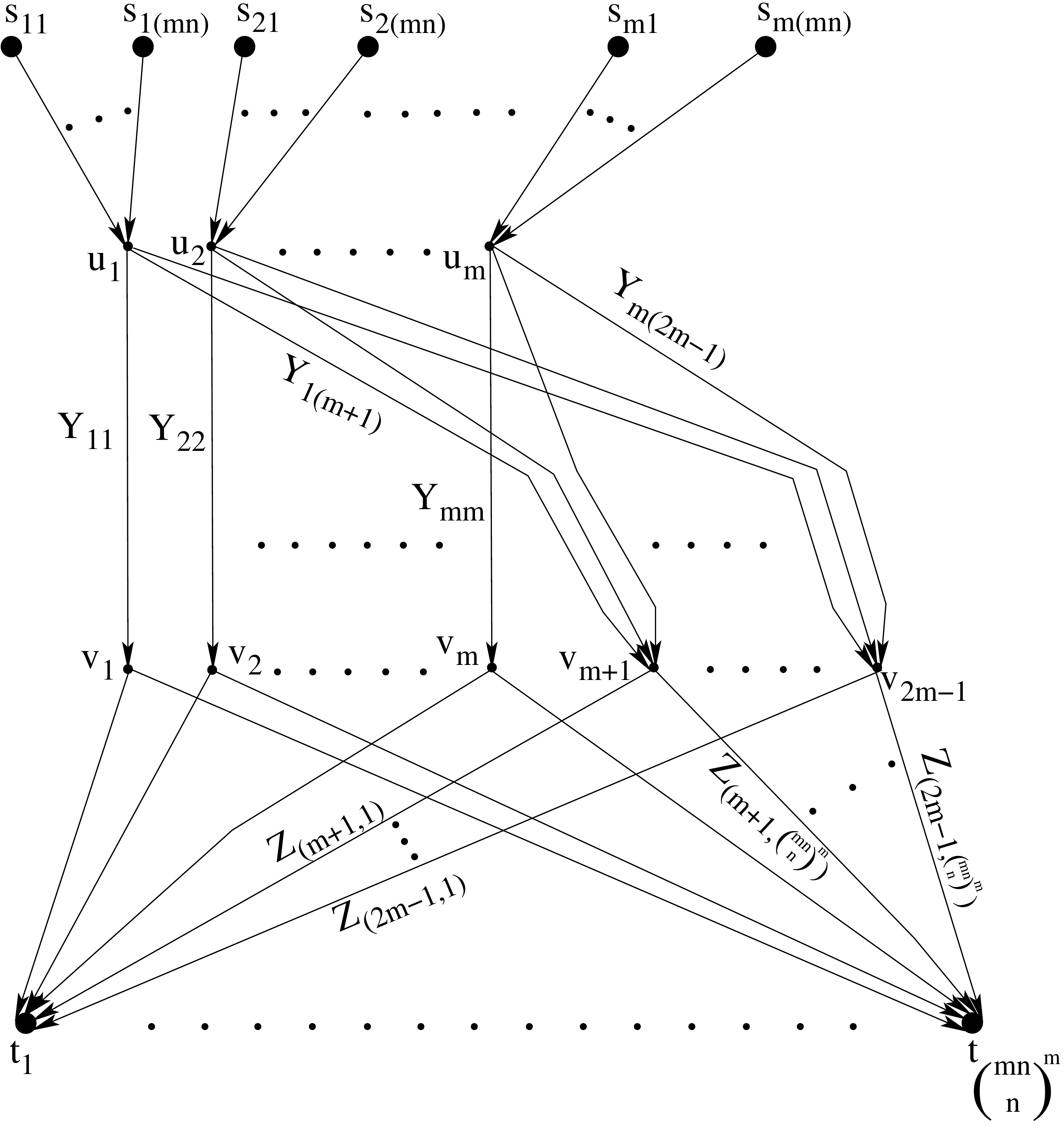}
\caption{A communication network $\mathcal{N}$. For any integer $m\geq 2$, $\mathcal{N}$ has an $(m,mn)$ fractional linear network code solution but has no $(w,wn)$ fractional linear network code solution when $w$ is less than $m$.}
\label{general}
\end{figure}
Consider the network $\mathcal{N}$ shown in Fig.~\ref{general}. We note that the network $\mathcal{N}$ is a further generalization of the ``generalized M-network'' presented in \cite{das}. $\mathcal{N}$ has $m^2n$ sources and $\binom {mn}{n}^m$ terminals. The sources are partitioned into $m$ sets $S_1,S_2,\ldots,S_m$. Each set has $mn$ many sources. The $j^{th}$ source in the set $S_i$ is denoted by $s_{ij}$. Each terminal demands information from $n$ sources of each set. If $T_{i}$ is the set of all sources demanded by $t_i$, then $T_{i}\neq T_{j}$ if $i\neq j$. Without loss of generality assume that the terminal $t_1$ demands the messages from the sources $s_{ij}$ for $1\leq i\leq m$, $1\leq j\leq n$. There are following sets of edges in the network $\mathcal{N}$: $\{(s_{ij},u_i)\} | 1\leq i\leq m, 1\leq j \leq mn$; $\{e_{ii} = (u_i,v_i)| 1\leq i\leq m\}$; $\{e_{ij} = (u_i,v_j) | 1\leq i\leq m, m+1\leq j\leq 2m-1\}$; and $\{(v_i$,$t_j) | 1\leq i\leq 2m-1, 1\leq j\leq \binom {mn}{n}^m\}$.
The message generated at the source $s_{ij}$ for $1\leq i\leq m$ and $1\leq j\leq mn$ is denoted by $X_{ij}$. All the source symbols are from the finite field $\mathbb{F}_q$. The information transmitted over the edge $e_{ii}$ for $1\leq i\leq m$ is denoted by $Y_{ii}$, the information transmitted over the edge $e_{ij}$ for $1\leq i\leq m$ and $m+1\leq j\leq 2m-1$ is denoted by $Y_{ij}$, and the information transmitted over the edge $(v_{i},t_{j})$ for $1\leq i\leq 2m-1$ and $1\leq j\leq \binom {mn}{n}^m$ is denoted by $Z_{ij}$. 

Note that for $n=1$, the network $\mathcal{N}$ reduces to the ``generalized M-network'' presented in \cite{das}.
\begin{lemma}\label{le1}
The capacity of the network $\mathcal{N}$ is upper bounded by $1/n$.
\end{lemma}
\textit{Proof.} Consider an $(r,l)$ fractional network code solution for $\mathcal{N}$. Since there is no source which is not demanded by any terminal, all the source messages in $S_1$ must get computed from the set $\{Y_{11},Y_{1(m+1)},Y_{1(m+2)},\ldots,Y_{1(2m-1)} \}$. Hence,
\begin{IEEEeqnarray*}{l}
H(Y_{11},Y_{1(m+1)},Y_{1(m+2)},\ldots ,Y_{1(2m-1)})\\
\qquad\qquad\qquad\qquad\qquad\qquad\geq H(X_{11},X_{12},\ldots ,X_{1(mn)})\\
\text{Or, } H(Y_{11}) + H(Y_{1(m+1)}) + \cdots + H(Y_{1(2m-1)})\\
\qquad\qquad\qquad\geq H(X_{11}) + H(X_{12}) + \cdots  + H(X_{1(mn)})\\
\text{Or, } ml\,\text{log}_2 q \geq r(mn)\text{log}_2 q\\
\text{Therefore, } \frac{r}{l} \leq \frac{1}{n} \qquad\qquad\qquad\qquad\qquad\qquad\qquad\qquad\qed
\end{IEEEeqnarray*}
 
Let us assume that the network $\mathcal{N}$ has a $(d,dn)$ fractional linear network code solution. The following theorem puts a constraint on $d$.

\begin{theorem}\label{thm1}
For an arbitrary finite field $\mathbb{F}_q$, there exists a network which has a $(d,dn)$ fractional linear network code solution over $\mathbb{F}_q$ if and only if $d$ is an integer multiple of $m$. 
\end{theorem}
We prove the theorem by using the network $\mathcal{N}$. We defer the proof of this theorem to Section~\ref{sec4}.
\begin{corollary}
The network $\mathcal{N}$ has an $(m,mn)$ fraction linear network code solution but has no $(w,wn)$ fractional linear network code solution for ${w<m}$.
\end{corollary}
For a special case of $n=1$, Theorem~\ref{thm1} specializes to the following Corollary:
\begin{corollary}
For an arbitrary finite field $\mathbb{F}_q$, there exists a network which has a $(d,d)$ fractional linear network code solution over $\mathbb{F}_q$ if and only if $d$ is an integer multiple of $m$. 
\end{corollary}
\noindent We note that this corollary is the main result of the paper \cite{das}.

Consider $k$ copies of $\mathcal{N}$, named as $\mathcal{N}_1,\mathcal{N}_2,\ldots,\mathcal{N}_k$. Denote the source $s_{ij}$ in $\mathcal{N}$ as the source $s_{ijp}$ in $\mathcal{N}_p$ for $1 \leq p \leq k$. Denote all the edges and terminals in $\mathcal{N}_p$ in a similar manner. We now construct a network $\mathcal{N}_{||}$ by connecting the $k$ copies of $\mathcal{N}$ in parallel such that for every $i$ and $j$ where $1 \leq i \leq m, 1 \leq j \leq mn$, $k$ sources $s_{ij1}$, $s_{ij2}$, $\ldots, s_{ijk}$ are combined and the combined source is denoted as $s_{ij}$. Similar combining is done with the terminals as well. Apart from sources and terminals, the rest of the nodes and edges of the $k$ copies of $\mathcal{N}$ remain disjoint in $\mathcal{N}_{||}$. 
\begin{theorem}
The capacity of $\mathcal{N}_{||}$ is upper bounded by $\frac{k}{n}$.
\end{theorem}
\begin{proof}
Consider an $(r,l)$ fractional linear network code solution for $\mathcal{N}_{||}$. Since from information carried by the edges in the set $\{\{e_{11p},e_{1(m+1)p},e_{1(m+2)p},\ldots,e_{1(2m-1)p}\}| 1\leq p \leq k\}$, $t_1$ must be able to compute all the messages from the sources in $S_1$, proceeding similar to Lemma~\ref{le1} we get $\frac{r}{l}\leq\frac{k}{n}$.
\end{proof}
Let $y$ be a positive rational number such that $yk$ and $yn$ are integers. We have the following lemma:
\begin{lemma}\label{l1}
If $\mathcal{N}_{||}$ has a $(yk,yn)$ fractional linear network code solution, then $\mathcal{N}$ has a $(yk,ykn)$ fractional linear network code solution. 
\end{lemma}
\begin{proof}
This is because, what can be sent from sources to the terminals in one unit of time in $\mathcal{N}_{||}$, the same can be sent in $k$ units of time in $\mathcal{N}$. 
\end{proof}
We now present the main result of this paper.
\begin{theorem}\label{thm2}
For any positive integers $k$, $n$, $y\geq 2$, and $w$, there exists a network which has a $(yk,yn)$ fractional linear network code solution, but has no $(wk,wn)$ fractional linear network code solution for all $w<y$. 
\end{theorem}
\begin{proof}
Consider $m=yk$ in $\mathcal{N}$. Say $\mathcal{N}_{||}$ has a $(wk,wn)$ fractional linear network code solution where $w<y$. Then, from Lemma~\ref{l1}, $\mathcal{N}$ has a $(wk,wkn)$ fractional linear network code solution. Therefore, from Theorem~\ref{thm1}, $m$ divides $wk$. However, since $m=yk$ and $wk$ is less than $yk$, $m$ cannot divide $wk$. This is a contradiction to the assumption that $\mathcal{N}_{||}$ has a $(wk,wn)$ fractional linear network code solution for $w<y$.

We now show that $\mathcal{N}_{||}$ has a $(yk,yn)$ fractional linear solution. Note that $\mathcal{N}$ has a $(yk,ykn)$ fractional linear network code solution. Theorem~\ref{thm1} ensures the existence of such a network code solution. Moreover, there is a $(yk,ykn)$ routing solution for $\mathcal{N}$ (see the proof of Theorem \ref{thm1} in Section \ref{sec4}). Therefore, trivially, $\mathcal{N}_{||}$ has a $(yk,yn)$ routing solution. This completes the proof. 
\end{proof}

\section{Proof of Theorem~\ref{thm1}}\label{sec4}
Our proof relies on a prior result on the connection between fractional linear network code solution and representable polymatroids. For the self-containment of the paper, we first provide a brief background material. For this, the definitions of a discrete polymatroid, and a discrete polymatroidal network are given in the following. These definitions are reproduced from \cite{sundar}. We also reproduce a prior result from \cite{sundar} which will be used in the proof of Theorem \ref{thm1}. We note the definitions and the result from \cite{sundar} are specialised to the network used in this paper. 

A discrete polymatroid is a set of vectors with some specific properties. Consider a set $G = \{1,2,\ldots,g\}$. Say $\mathbb{Z}_{\geq 0}$ denote the set of all positive integers. Also let $\mathbb{Z}_{\geq 0}^g$ be the set of all $g$ length vectors whose components belong to $\mathbb{Z}_{\geq 0}$. If $x$ is a vector then let $x(i)$ denote the $i^{th}$ component of $x$. A function $\rho : 2^G \rightarrow \mathbb{Z}_{\geq 0}$ qualifies to be a rank function if the following three conditions are satisfied:\\
$(1)$ $\rho(\phi) = 0 $\\
$(2)$ if $A\subseteq B\subseteq G$, then $\rho(A)\leq \rho(B)$\\
$(3)$ $\rho(A) + \rho(B) \geq \rho(A\cup B) + \rho(A\cap B)$ for $\forall A,B \subseteq G$.\\
A discrete polymatroid $\mathbb{D}$ with rank function $\rho$ is defined as $\mathbb{D} = \{x \in \mathbb{Z}_{\geq 0}^g | \sum_{i\in A} x(i) \leq \rho(A), \forall A \subseteq G\}$.
%
$\mathbb{D}$ is said to be representable over $\mathbb{F}_q$ if there exist vector subspaces $V_1,V_2,\ldots,V_g$ of a vector space $E$ over $\mathbb{F}_q$ such that for any $A\subseteq G$, $dim(\sum_{i\in A} V_i) = \rho(A)$.\\
We now state the conditions for $\mathcal{N}$ to be an $(r,l)$ discrete polymatroidal network. Recall that $X_{ij}$ is the messages generated by the source $s_{ij}\in S$. Say $X = \{X_{ij} | 1\leq i\leq m, 1\leq j\leq mn \}$. Let $Y$ be the collection of all messages carried by all the edges of the network. For any node $v$, let $Y_{In(v)}$ and $Y_{Out(v)}$ denote the set of messages carried the edges in $In(v)$ and $Out(v)$ respectively.\\ The network $\mathcal{N}$ is an $(r_1,r_2,\ldots,r_{m^2n};l)$-discrete polymatroidal network with respect to the discrete polymatroid $\mathbb{D}$ if there exists a map $f: \{X \cup Y\} \rightarrow G$ which satisfies the following conditions:\\
$(1)$ $f$ is one-to-one on the elements of $X$\\
$(2)$ $\rho(f(X_{ij})) = r_{(i-1)mn+j}$ and $\rho(f(Y_e)) \leq l$ for any $X_{ij}\in X$, $Y_e\in Y$\\
$(3)$ $\sum r_i\epsilon_{ig} \in \mathbb{D}$ where $\epsilon_{ig}\in \mathbb{Z}_{\geq 0}^g$, $\epsilon_{ig}(i) = 1$ and $\epsilon_{ig}(j) = 0$ if $j\neq i$ for $1\leq i\leq m^2n$\\
$(4)$ For any node $v$, $\rho(f(Y_{In(v)})) = \rho(f(Y_{In(v)}\cup Y_{Out(v)}))$.\vspace*{1ex}
The network $\mathcal{N}$ is an $(r,l)$-discrete polymatroidal network if it is an $(r_1,r_2,\ldots,r_{m^2n};l)$-discrete polymatroidal network for $r_1,r_2,\ldots,r_{m^2n} = r$.\\
The following theorem is reproduced with slightly different notations from \cite{sundar}.
\begin{theorem}
A network has an $(r_1,r_2,\ldots,r_{m^2n};l)$ fractional linear network code solution over $\mathbb{F}_q$, if and only if it is an $(r_1,r_2,\ldots,r_{m^2n};l)$-discrete polymatroidal with respect to a discrete polymatroid $\mathbb{D}$ representable over $\mathbb{F}_q$.
\end{theorem}
An $(r,l)$ fractional linear network code solution is an $(r_1,r_2,\ldots,r_{m^2n};l)$ fractional linear network code solution for $r_1,r_2,\ldots,r_{m^2n} = r$.\\
\textit{Proof of Theorem~\ref{thm1}:} We now show that the network $\mathcal{N}$ has a $(d,dn)$ fractional linear network code solution if and only if $d$ is an integer multiple of $m$. We first prove the `only if' part by showing that the network $\mathcal{N}$ is a $(d;dn)$-discrete polymatroidal network only if $m$ divides $d$. We then prove the `if' part by presenting an $(m,mn)$ fractional linear network code solution.

Let $\mathtt{g} = \rho\circ f$. Also let $\mathbb{Z}_n = \{0,1,\ldots,n-1\}$. Define the set $U_i = \{a + 1 + (i-1)n\; | a\in \mathbb{Z}_n \}$ and the set $C_{ij} = \{X_{ia} | a\in U_j\}$ for $1\leq i,j\leq m$.
\begin{claim}\label{clai}
For any $X_{Q_i}\subseteq \{X_{i1},X_{i2},\ldots ,X_{i(mn)}\}$,
\begin{IEEEeqnarray*}{l}
\mathtt{g}(Y_{11},X_{Q_{1}}) + \mathtt{g}(Y_{22},X_{Q_2}) + \cdots + \mathtt{g}(Y_{mm},X_{Q_m})\\
= \mathtt{g}(Y_{11},X_{Q_1},Y_{22},X_{Q_2},\ldots ,Y_{mm},X_{Q_m}).
\end{IEEEeqnarray*}
\end{claim}
\begin{proof}
Since none of the vectors in $X_{Q_i}$ and $Y_{ii}$ are dependent on any vector in $X_{Q_j}$ and $Y_{jj}$ for $i\neq j$, $1\leq i,j\leq m$, we have $H(Y_{ii},X_{Q_i}|Y_{jj},X_{Q_j}) = H(Y_{ii},X_{Q_i})$. Hence, $H(Y_{ii},X_{Q_i}) + H(Y_{jj},X_{Q_j}) = H(Y_{ii},X_{Q_i},Y_{jj},X_{Q_j})$. Similarly the following equation must also be true:
\begin{IEEEeqnarray*}{l}
H(Y_{11},X_{Q_{1}}) + H(Y_{22},X_{Q_2}) + \cdots + H(Y_{mm},X_{Q_m})\\
= H(Y_{11},X_{Q_1},Y_{22},X_{Q_2},\ldots ,Y_{mm},X_{Q_m})\IEEEyesnumber\label{ie1}
\end{IEEEeqnarray*}
Since $\mathtt{g}(\,)$ is a rank function of a discrete polymatroid, the equation~(\ref{ie1}) remains valid if $H(\,)$ is replaced by $\mathtt{g}(\,)$ \cite{dougherty}.
\end{proof}

\begin{claim}\label{cla1}
Let $X_{P_i}\!\subset \{X_{i1},X_{i2},\ldots ,X_{i(mn)}\}$, and $|X_{P_i}|=n$ for $1\leq i\leq m$. The following set of inequalities hold.
\begin{multline}
\mathtt{g}(Y_{11},X_{P_{1}}) + \mathtt{g}(Y_{22},X_{P_2}) + \cdots + \mathtt{g}(Y_{mm},X_{P_m})\\ \leq (2m-1)nd \label{mq1}
\end{multline}
\end{claim}
\begin{proof}
We show the proof for the case when $X_{P_i} = C_{i1}$ for $1\leq i \leq m$. The rest of the cases can also be proved similarly.
\begin{IEEEeqnarray*}{l}
\mathtt{g}(Y_{11},X_{P_{1}}) + \mathtt{g}(Y_{22},X_{P_2}) + \cdots + \mathtt{g}(Y_{mm},X_{P_m})\\
= \mathtt{g}(Y_{11},X_{P_1},Y_{22},X_{P_2},\ldots ,Y_{mm},X_{P_m})\IEEEyesnumber\label{nq1}\\
\leq \mathtt{g}(Y_{11},X_{P_1},Y_{22},X_{P_2},\ldots ,Y_{mm},X_{P_m},\\
\quad Z_{(m+1,1)},\ldots ,Z_{(2m-1,1)})\\
= \mathtt{g}(Y_{11}\ldots ,Y_{mm},Z_{(m+1,1)},\ldots ,Z_{(2m-1,1)})\IEEEyesnumber\label{nq2}\\
\leq mdn + (m-1)dn\\
= (2m-1)dn
\end{IEEEeqnarray*}
Equation~(\ref{nq1}) comes from Claim~\ref{clai}. Equation~(\ref{nq2}) is true because the terminal $t_1$ computes all the source symbols in the set $\{\cup_{1\leq i\leq m} C_{i1}\}$ from the vectors in the set $\{Y_{11}\ldots ,Y_{mm},Z_{(m+1,1)},\ldots,Z_{(2m-1,1)}\}$. Note that for all configuration of $X_{P_1},X_{P_2},\ldots ,X_{P_m}$ there is a terminal which demands all the sources contained in these sets. So the rest of the inequalities can be proved analogously.
\end{proof}
\begin{claim}\label{cla2}
For $1\leq i\leq m$ the set $\{X_{i1},X_{i2},\ldots ,X_{i(mn)}\}$ is arbitrarily partitioned into $m$ mutually disjoint sets each containing $n$ elements. Denote the $j^{th}$ set as $X_{U_{ij}}$. 
%
The following inequalities hold true. 
\begin{multline}
\text{For } 1\leq i\leq m,\\
\mathtt{g}(Y_{ii},X_{U_{i1}}) + \mathtt{g}(Y_{ii},X_{U_{i2}}) + \cdots + \mathtt{g}(Y_{ii},X_{U_{im}}) \\ \geq  (2m-1)nd \label{claeq}
\end{multline}
\end{claim}

\begin{proof}
To prove this result we first show that $\mathtt{g}(Y_{ii}) = \mathtt{g}(Y_{ij}) = nd$ for $1\leq i\leq m$ and $(m+1)\leq j\leq (2m-1)$.
\begin{IEEEeqnarray*}{ll}
m^2nd&\\
=\mathtt{g}(X_{11},X_{12},\ldots ,X_{1(mq)},&\ldots ,X_{m(mn)})\\
\leq \mathtt{g}(X_{11},X_{12},\ldots ,X_{1(mn)},&\ldots ,X_{m(mn)},Y_{11},\ldots ,Y_{mm},\\&Y_{1(m+1)},\ldots ,Y_{m(2m-1)})\\
= \mathtt{g}(Y_{11},\ldots ,Y_{mm},Y_{1(m+1)}&,\ldots ,Y_{m(2m-1)})\IEEEyesnumber\label{nq3}\\
\leq \mathtt{g}(Y_{11}) + \cdots + \mathtt{g}(Y_{mm}) + &\mathtt{g}(Y_{1(m+1)}) + \cdots + \mathtt{g}(Y_{m(2m-1)})\\
\leq mnd + m(m-1)nd&\IEEEyesnumber\label{nq4}\\
= m^2nd& 
\end{IEEEeqnarray*}
Equation~(\ref{nq3}) is true as there exists no source symbol which is not demanded by any terminal. Equation~(\ref{nq4}) comes from the fact that $\mathtt{g}(Y_{ii}) \leq nd$ and $\mathtt{g}(Y_{ij}) \leq nd$ for $1\leq i\leq m$ and \mbox{$(m+1)\leq j\leq (2m-1)$}. So $\mathtt{g}(Y_{11}) + \cdots + \mathtt{g}(Y_{mm}) + \mathtt{g}(Y_{1(m+1)}) + \cdots + \mathtt{g}(Y_{m(2m-1)}) = m^2nd$. This implies $\mathtt{g}(Y_{ii}) = \mathtt{g}(Y_{ij}) = nd$. \\

We now prove equation~(\ref{claeq}) when $X_{U_{ij}} = C_{ij}$ for $1\leq j\leq m$.
\begin{IEEEeqnarray*}{l}
\mathtt{g}(Y_{ii},C_{i1}) + \mathtt{g}(Y_{ii},C_{i2}) + \cdots + \mathtt{g}(Y_{ii},C_{im})\\
\geq \mathtt{g}(Y_{ii}, C_{i1}, C_{i2},\ldots ,C_{im}) + (m-1)Y_{ii}\\
= \mathtt{g}(C_{i1}, C_{i2},\ldots ,C_{im}) + (m-1)Y_{ii}\IEEEyesnumber\label{hq1}\\
= mnd + (m-1)nd \IEEEyesnumber\label{hq2}\\
= (2m-1)nd 
\end{IEEEeqnarray*}
The inequality in equation~(\ref{hq1}) is obtained by using the n-way submodularity given in \cite{rasala}. Equation~(\ref{hq2}) is true because $Y_{ii}$ can be computed from $C_{i1}, C_{i2},\ldots ,C_{im}$. For all other possible configurations of the sets $X_{U_{i1}},X_{U_{i2}},\ldots ,X_{U_{im}}$ the equations can be proved similarly.
\end{proof}

\begin{claim}\label{cla3}
If $X_{R}\subset \{X_{i1},X_{i2},\ldots,X_{i(mn)}\}$ and $|X_{R}| = n$, then for $1\leq i\leq m$, $\mathtt{g}(Y_{ii},X_{R}) \leq \frac{(2m-1)nd}{m}$.
\end{claim}

\begin{proof}
We show the proof of $\mathtt{g}(Y_{ii},X_{R}) \leq \frac{(2m-1)nd}{m}$ for $X_{R} = \{X_{m1},X_{m2},\ldots,X_{mn}\}$. For other possibilities of $X_{R}$ the claim can be proved similarly. Consider the following $m$ inequalities from equation~(\ref{mq1}) obtained by taking $X_{P_1} \in \{C_{11},C_{12},\ldots,C_{1m}\}$, and $X_{P_i} = C_{i1}$ for $2\leq i\leq m$.
\begin{IEEEeqnarray*}{l}
\mathtt{g}(Y_{11},C_{11}) {+} \mathtt{g}(Y_{22},C_{21}) {+} \cdots {+} \mathtt{g}(Y_{mm},C_{m1}) {\leq} (2m-1)nd\\
\mathtt{g}(Y_{11},C_{12}) {+} \mathtt{g}(Y_{22},C_{21}) {+} \cdots {+} \mathtt{g}(Y_{mm},C_{m1}) {\leq} (2m-1)nd\\
\qquad \qquad \quad\textsf{:}\qquad \qquad \quad  \textsf{:}  \qquad  \textsf{:}\qquad \qquad \qquad \quad \textsf{:}\\
\mathtt{g}(Y_{11},C_{1m}) {+} \mathtt{g}(Y_{22},C_{21}) {+} {\cdots} {+} \mathtt{g}(Y_{mm},C_{m1}) {\leq} (2m-1)nd
\end{IEEEeqnarray*}
By summing up these $m$ equations, we get:
\begin{multline}
\mathtt{g}(Y_{11},C_{11}) + \mathtt{g}(Y_{11},C_{12}) + \cdots + \mathtt{g}(Y_{11},C_{1m}) + \\m\{\mathtt{g}(Y_{22},C_{21}) + \mathtt{g}(Y_{33},C_{31}) + \cdots + \mathtt{g}(Y_{mm},C_{m1})\}\\ \leq m(2m-1)nd \label{www}
\end{multline}
From Claim~\ref{cla2}, we have $\mathtt{g}(Y_{11},C_{11}) + \mathtt{g}(Y_{11},C_{12}) + \cdots + \mathtt{g}(Y_{11},C_{1m}) \geq (2m-1)nd$. Substituting this in equation~(\ref{www}), we get:
\begin{multline}
m\{\mathtt{g}(Y_{22},C_{21}) + \mathtt{g}(Y_{33},C_{31}) + \cdots + \mathtt{g}(Y_{mm},C_{m1}) \} \\ \leq m(2m-1)nd - (2m-1)nd
\end{multline}
Similarly, considering the inequalities from equation~(\ref{mq1}) obtained by taking $X_{P_1} \in \{C_{11},C_{12},\ldots,C_{1m}\}$, $X_{P_2} = C_{22}$, and $X_{P_i} = C_{i1}$ for $3\leq i\leq m$; and then summing the $m$ equations, we can get the following:
\begin{multline}
m\{\mathtt{g}(Y_{22},C_{22}) + \mathtt{g}(Y_{33},C_{31}) + \cdots + \mathtt{g}(Y_{mm},C_{m1}) \} \\ \leq m(2m-1)nd - (2m-1)nd
\end{multline}
In the same manner it can be seen that for $1\leq i\leq m$,
\begin{multline}
m\{\mathtt{g}(Y_{22},C_{2i}) + \mathtt{g}(Y_{33},C_{31}) + \cdots + \mathtt{g}(Y_{mm},C_{m1}) \} \\ \leq m(2m-1)nd - (2m-1)nd \label{mu1}
\end{multline}
Note that in these $m$ equations, there is no term involving $\mathtt{g}(Y_{11},X_{P_1})$ for any $X_{P_1}$. Summing the $m$ equations in equation~(\ref{mu1}) we get:
\begin{multline}
m\{\mathtt{g}(Y_{22},C_{21}) + \mathtt{g}(Y_{22},C_{22}) + \cdots + \mathtt{g}(Y_{22},C_{2m}) \} + \\m^2\{\mathtt{g}(Y_{33},C_{31}) + \cdots + \mathtt{g}(Y_{mm},C_{m1}) \} \\ \leq m^2(2m-1)nd - m(2m-1)nd \label{www1}
\end{multline}
From equation~(\ref{claeq}), we have $\mathtt{g}(Y_{22},C_{21}) + \mathtt{g}(Y_{22},C_{22}) + \cdots + \mathtt{g}(Y_{22},C_{2m}) \geq (2m-1)nd$. Substituting this in equation~(\ref{www1}) we get:
\begin{multline}
m^2\{\mathtt{g}(Y_{33},C_{31}) + \cdots + \mathtt{g}(Y_{mm},C_{m1}) \} \\ \leq m^2(2m-1)nd - 2m(2m-1)nd
\end{multline}
Note again that the last equation has no terms involving either $\mathtt{g}(Y_{11},X_{P_1})$ or $\mathtt{g}(Y_{22},X_{P_2})$ for any $X_{P_1}$ and $X_{P_2}$. In this way after eliminating all terms involving $\mathtt{g}(Y_{ii},X_{P_i})$ for $1\leq i\leq m-1$, we get the following equation:
\begin{IEEEeqnarray*}{l}
m^{m-1}\mathtt{g}(Y_{mm},C_{m1}) \\ \leq m^{m-1}(2m-1)nd - (m-1)m^{m-2}(2m-1)nd\\
\Rightarrow \mathtt{g}(Y_{mm},C_{m1})\leq \frac{(2m-1)nd}{m}
\end{IEEEeqnarray*}

Now noting that $X_{R} = \{X_{m1},X_{m2},\ldots,X_{mn}\} = C_{m1}$,
\begin{IEEEeqnarray*}{l}
\mathtt{g}(Y_{mm},X_{R}) \leq \frac{(2m-1)nd}{m}
\end{IEEEeqnarray*}
\end{proof}
\begin{claim}\label{cla4}
For $1\leq i\leq m$, if $X_{V_i} \subset \{ X_{i1},X_{i2},\ldots,X_{i(mn)}\}$ and $|X_{V_i}| = y$, then for $1\leq y \leq n$, $\mathtt{g}(Y_{ii},X_{V_i}) \leq \frac{(nm + ym -y)d}{m}$.
\end{claim}
\begin{proof}
From Claim~(\ref{cla3}) it can be seen that the result is true for $y=n$. We show that if the result is true for an arbitrary value of $y = z$ $(2\leq z\leq n)$, it is also true for $y = z-1$. The proof for the case when $X_{V_i} = \{X_{i1},X_{i2},\ldots,X_{i(z-1)} \}$ is given below assuming that for $\forall y\geq z$, $\mathtt{g}(Y_{ii},X_{V_i}) \leq \frac{(nm + ym -y)d}{m}$ when $|X_{V_i}| = y$. For the rest of the possibilities of $X_{V_i}$ when $|X_{V_i}| = z-1$, the proof is similar.
\begin{IEEEeqnarray*}{l}
\mathtt{g}(Y_{ii},X_{V_i},X_{iz}) + \mathtt{g}(Y_{ii},X_{V_i},X_{i(z+1)}) + \cdots +\\ \mathtt{g}(Y_{ii},X_{V_i},X_{i(mn)})
\geq  \mathtt{g}(Y_{ii},X_{V_i},X_{iz},X_{i(z+1)},\ldots,X_{i(mn)}) \\  \qquad \qquad \qquad \qquad \qquad \qquad  +\> (mn-z) \mathtt{g}(Y_{ii},X_{V_i})\IEEEyesnumber\label{a11}\\
\text{or, } (mn{-}z{+}1)\frac{(nm {+} zm {-}z)d}{m} \geq mnd {+} (mn{-}z) \mathtt{g}(Y_{ii},X_{V_i})\\
\text{or, } (mn{-}z)\frac{(nm {+} zm {-}z)d}{m} + \frac{(nm {+} zm {-}z)d}{m} - mnd \\  \qquad \qquad \qquad \qquad \qquad \qquad \qquad \geq (mn{-}z)   \mathtt{g}(Y_{ii},X_{V_i})\\
\text{or, } (mn{-}z)\frac{(nm {+} zm {-}z {-}m {+} m {-}1 {+} 1)d}{m}  + \frac{(nm {+} zm {-}z)d}{m} \\ \qquad \qquad \qquad \qquad\qquad \qquad -\> mnd \geq (mn{-}z)   \mathtt{g}(Y_{ii},X_{V_i})\\
\text{or, } (mn{-}z)\frac{(nm {+} (z{-}1)m {-}(z{-}1) {+} m {-}1)d}{m}  \\ \qquad \qquad  -\> \frac{(m^2n { -}nm {-} zm {+}z)d}{m}  \geq (mn{-}z)   \mathtt{g}(Y_{ii},X_{V_i})\\
\text{or, } (mn{-}z)\frac{(nm {+} (z{-}1)m {-}(z{-}1))d}{m} + (mn{-}z)\frac{(m {-}1)d}{m} \\ \qquad \qquad  -\> \frac{(mn(m{-}1){-}z(m{-}1))d}{m}  \geq (mn{-}z)   \mathtt{g}(Y_{ii},X_{V_i})\\
\text{or, } (mn{-}z)\frac{(nm {+} (z{-}1)m {-}(z{-}1))d}{m} \geq (mn{-}z)   \mathtt{g}(Y_{ii},X_{V_i})\\
\text{or, } \frac{(nm {+} (z{-}1)m {-}(z{-}1))d}{m} \geq  \mathtt{g}(Y_{ii},X_{V_i})
\end{IEEEeqnarray*}
Equation~(\ref{a11}) comes from the n-way submodularity formula given in \cite{rasala}.
\end{proof}
For $y=1$ in Claim~(\ref{cla4}), we get: $\mathtt{g}(Y_{ii},X_{ij}) \leq \frac{(nm + m - 1)d}{m}$ for $1\leq i\leq m$, $1\leq j\leq mq$.\\Now using the n-way submodularity again,
\begin{IEEEeqnarray*}{l}
\mathtt{g}(Y_{ii},X_{i1}) + \mathtt{g}(Y_{ii},X_{i2}) + \cdots + \mathtt{g}(Y_{ii},X_{i(mn)})\\
\geq \mathtt{g}(Y_{ii},X_{i1},X_{i2},\ldots,X_{i(mn)}) + (mn-1)\mathtt{g}(Y_{ii})\\
\text{or, } \mathtt{g}(Y_{ii},X_{i1}) + \mathtt{g}(Y_{ii},X_{i2}) + \cdots + \mathtt{g}(Y_{ii},X_{i(mn)})
\\ \geq mnd + (mn-1)nd = (nm + m -1)nd \IEEEyesnumber\label{a22}
\end{IEEEeqnarray*}
Since there are $mn$ terms in the right hand side of equation~(\ref{a22}) and each term is less than or equal to $\frac{(nm + m - 1)d}{m}$, it must be that $\mathtt{g}(Y_{ii},X_{ij}) = \frac{(nm + m - 1)d}{m}$. Now note that $gcd(nm+m-1,m) = gcd(-1,m) = gcd(m-1,m) = gcd(m,1) = 1$. Since by definition, the rank function of a discrete polymatroid $\mathtt{g}(\,)$ is always an integer, $m$ must divide $d$ for $\mathtt{g}(Y_{ii},X_{ij})$ to be an integer.

We now show that the network $\mathcal{N}$ has an $(m,mn)$ fractional linear network code solution over any finite field by presenting an $(m,mn)$ routing solution. Note that an $(m,mn)$ routing solution is a special case of an $(m,mn)$ fractional linear network code solution. For any vector $X_{ij}$, $1\leq i\leq m$, $1\leq j\leq mn$, the first component is carried by $e_{ii}$, and the $p^{th}$ component for $2\leq p\leq m$ is carried by $e_{i(m-1+p)}$. Now the demands of all the terminals can be met by routing the appropriate symbols from the node $v_{i}$ for $1\leq i\leq 2m-1$. For example, the demands of the terminal $t_1$ gets fulfilled upon receiving the appropriate $n$ symbols of all vectors $Y_{ii}$ and $Y_{ij}$ for $1\leq i\leq m$, $m+1\leq j\leq 2m-1$. \qed

\section{Conclusion}\label{conclusion}
For non-multicast networks, it was shown that there exists networks which has an $(m,m)$ $(m\geq 2)$ vector linear network code solution, but has no $(w,w)$ vector linear network code solution for $w<m$. In this paper, we have generalized this result to show that for any positive integers $k$ and $n$, there exists a network which has no $(wk,wn)$ fractional linear network code solution for any $w<m$, but has an $(mk,mn)$ $(m\geq 2)$ fractional linear network code solution.

\end{document}